\documentclass[12pt,a4paper]{article}
\usepackage[utf8]{inputenc}
\usepackage[english]{babel}
\usepackage{amsmath}
\usepackage{amsfonts}
\usepackage{amssymb}
\usepackage[margin=1in]{geometry}
\usepackage{amsthm}
\usepackage{hyperref}
\usepackage{graphicx}
\usepackage[numbers]{natbib}

\newtheorem{as}{Assumption}[section]
\newtheorem{df}[as]{Definition}

\newtheorem{prp}[as]{Proposition}
\newtheorem{thm}[as]{Theorem}

\newtheorem{rem}[as]{Remark}

\numberwithin{equation}{section}
\numberwithin{figure}{section}
\numberwithin{table}{section}

\title{Optimal Investment with Stochastic Interest Rates and Ambiguity}
\author{Julian H\"olzermann\footnote{University of Southern Denmark and Danish Finance Institute. Email: juho@sam.sdu.dk. The author thanks Steve Bell, Christian Riis Flor, Nikolaus Schweizer, Hoi Ying Wong, and the participants of the 11th General AMaMeF Conference and the Economic Theory Lunch Seminar at Bielefeld University for comments and suggestions.}}

\begin{document}

\maketitle

\begin{abstract}
\noindent This paper studies dynamic asset allocation with interest rate risk and several sources of ambiguity. The market consists of a risk-free asset, a zero-coupon bond (both determined by a Vasicek model), and a stock. There is ambiguity about the risk premia, the volatilities, and the correlation. The investor's preferences display both risk aversion and ambiguity aversion. The optimal investment problem admits a closed-form solution. The solution shows that the ambiguity only affects the speculative motives of the investor, representing a hedge against the ambiguity, but not the hedging of interest rate risk. An implementation of the optimal investment strategy shows that ambiguity aversion helps to tame the highly leveraged portfolios neglecting ambiguity and leads to strategies that are more in line with popular investment advice.
\end{abstract}

\noindent\textbf{Keywords:} Portfolio Selection, Robust Investment, Bond-Stock Ratio, Fixed Income Management, Knightian Uncertainty, Model Uncertainty
\\\textbf{JEL Classification:} G11, G12
\\\textbf{MSC2020:} 91G10, 91G30

\section{Introduction}
How to optimally allocate wealth between bonds, stocks, and cash when there is ambiguity? Without ambiguity, this is an important example of an asset allocation problem, which is typically addressed by dynamic asset allocation models with interest rate risk. Such models show how rational investors can deal with the risk induced by (long-term) bonds and stocks as well as the risk related to the short-term interest rate. However, the models do not show how to deal with ambiguity. Ambiguity refers to the uncertainty that cannot be described by probabilities and is surely present in financial markets, since model parameters are far from perfectly known and often hard to estimate. Therefore, it is natural to treat the parameters of the asset allocation model as ambiguous and study how ambiguity affects the optimal portfolio choice between bonds, stocks, and cash.
\par This is captured by an optimal investment problem with several sources of ambiguity. The starting point is a setting purely based on risk, similar to (classical) dynamic asset allocation models with interest rate risk. The market offers a locally risk-free asset with a risky short-term interest rate. In addition, there is a zero-coupon bond with deterministic risk premium and volatility. The stock market is represented by a single risky asset with constant risk premium, volatility, and correlation (between the stock and the short-term interest rate as well as the bond). The investor chooses how much wealth to invest in the different assets, which can be continuously rebalanced over time. Since there is ambiguity, the investor does not only consider a single scenario for the risk premia, the volatilities, and the correlation but a collection of possible scenarios without any assumptions on which is more likely to be the correct one. Each scenario leads to a different evolution of the investor's wealth and the short-term interest rate. As it is common in the literature, the ambiguity is represented by a set of priors \citep{epsteinji2013,epsteinji2014}, where each prior represents a specific scenario for the risk premia, the volatilities, and the correlation. The investor's preferences display both risk aversion and ambiguity aversion, represented by a constant relative risk aversion utility function and maxmin expected utility in the spirit of \citet{gilboaschmeidler1989}, respectively. Thus, the investor aims at maximizing expected utility under the worst possible prior.
\par The solution to the optimal investment problem shows how the ambiguity affects the optimal investment strategy of the investor. The optimal investment problem can be solved in closed form for the optimal investment strategy, the worst-case prior, and the value function related to the problem, constituting the main result of the paper. Since the optimal investment strategy can be decomposed into a speculative part and a part that hedges interest rate risk (similar to the case without ambiguity), the solution shows that only the speculative part is altered by the ambiguity, while the hedge part is not. The change in the speculative part of the optimal investment strategy compared to the optimal investment strategy neglecting ambiguity can be interpreted as a hedge against the ambiguity. In particular, the ambiguity about the correlation and the worst-case Sharpe ratios of the bond and the stock determine whether the investor invests in the bond or the stock for speculative reasons, while the investor always invests in the bond to hedge interest rate risk. In contrast to the case without ambiguity, the ambiguous risk premium on the bond, which is endogenously determined by the ambiguous volatility of the short-term interest rate, leads to an additional dynamic effect in the optimal investment strategy.
\par An implementation of the optimal investment strategy shows the quantitative effects of ambiguity aversion. The calibration is similar to the case without ambiguity but estimates the extreme values for the ambiguous quantities by the extreme values of rolling window estimates, which is reasonable in a dynamic setting since the ambiguous quantities are allowed to vary over time, but there are in general also other ways to estimate parameters for the ambiguous quantities, e.g., by confidence intervals, which is suitable for a one-period framework, where the ambiguous quantities do not vary over time \citep{danglweissensteiner2020,garlappietal2007}. The estimates yield the optimal investment strategy for an ambiguity averse investor and the optimal investment strategy for an investor neglecting ambiguity. Comparing both strategies shows that ambiguity aversion leads to more reasonable and less leveraged strategies, which still display the desirable properties from the case without ambiguity. By looking at the different sources of ambiguity separately, one can see how each source influences the optimal investment strategy. Moreover, the implementation shows that the optimal investment strategy with ambiguity aversion has a much more natural behavior over time, which is more in line with typical investment advice in contrast to the optimal investment strategy neglecting ambiguity. The reason for the change in the dynamic behavior is the additional dynamic effect in the optimal investment strategy, which is due to the ambiguous volatility of the short-term interest rate.
\par The approach to solving the optimal investment problem is based on an extension of the martingale optimality principle. The martingale optimality principle provides a very simple method for solving dynamic asset allocation problems in closed-form \citep{korn2003,rogers2013}, although in applications, it is very similar to the classical dynamic programming approach to stochastic control problems. But more importantly, it is straightforward to extend the martingale optimality principle to a setting with multiple priors, that is, settings incorporating ambiguity \citep{biaginipinar2017,linriedel2021}. Since the focus of the present paper is on obtaining a closed-form solution to an optimal investment problem in the presence of ambiguity, the martingale optimality principle is tailor-made for solving the problem.
\par Allocating wealth between bonds, stocks, and cash plays an important role in the (dynamic) asset allocation literature. The reason is that according to typical investment advice, more risk averse investors should invest more in bonds compared to stocks, which seems intuitive but contradicts classical fund separation results, as pointed out by \citet{canneretal1997}. A solution to this problem is to consider dynamic asset allocation models with interest rate risk, as it is done by \citet{bajeux-besnainouetal2001,bajeux-besnainouetal2003}, \citet{brennanxia2000}, \citet{sorensen1999}, and \citet{wachter2003} in a continuous-time framework and by \citet{campbellviceira2001} in a discrete-time framework. Due to the importance of the results, there are many papers on this problem in the asset allocation literature: \citet{liouiponcet2001} and \citet{munksorensen2004} extend the results to more general settings, \citet{kornkraft2001,kornkraft2004} provide a detailed analysis of some technical issues related to the problem, and \citet{brangeretal2023} study collective investment in bonds, stocks, and cash (to name a few).
\par In addition to risk, there is a stream of literature on asset allocation focusing on ambiguity, which the present paper complements. \citet{garlappietal2007} introduce ambiguity to the classical mean-variance analysis, and \citet{danglweissensteiner2020} follow a similar approach to consider the allocation between bonds, stocks, and cash (and other asset classes). Instead of a one-period setting, one can also introduce ambiguity to continuous-time portfolio selection problems, as it is done by \citet{brangeretal2013} and \citet{maenhout2004,maenhout2006} using a maxmin approach and by \citet{balteretal2021} using the smooth ambiguity approach. Similar to the present paper, \citet{florlarsen2014} study optimal investment in bonds, stocks, and cash with ambiguity. However, as the previous studies on ambiguity in continuous time, they construct the set of priors using equivalent probability measures, which restricts the ambiguity to the drift of the underlying processes, that is, ambiguous risk premia. In addition to ambiguous risk premia, it is important to consider ambiguous volatility. \citet{epsteinji2013,epsteinji2014} offer a detailed discussion about the relevance of ambiguous volatility and study asset pricing. \citet{kostopoulosetal2022} provide empirical support for the effect of ambiguous volatility on investor behavior. \citet{biaginipinar2017} and \citet{linriedel2021} study dynamic asset allocation problems with several sources of ambiguity--including ambiguous volatility. In addition to ambiguous volatility, \citet{epsteinhalevy2019} highlight the role of ambiguous correlation and provide experimental evidence. \citet{fouqueetal2016} and \citet{linetal2022} examine asset allocation with ambiguous correlation. However, none of the studies with additional sources of ambiguity (in addition to ambiguous risk premia) focuses on the allocation between bonds, stocks, and cash. Therefore, the motivation of the present paper is to investigate the asset allocation problem between bonds, stocks, and cash with ambiguous risk premia, ambiguous volatilities, and ambiguous correlation.
\par The paper is organized as follows. Section \ref{investment problem} introduces the investment problem the investor faces. Section \ref{optimal investment} states and discusses the main result: the optimal investment strategy. Section \ref{implementation} implements the optimal investment strategy and offers quantitative results. Section \ref{conclusion} gives a conclusion. The proof of the main result and some details related to the calibration of the model are deferred to Sections \ref{proof} and \ref{calibration negelcting ambiguity} in the appendix, respectively.

\section{Investment Problem}\label{investment problem}
All risk in the model is represented by a two-dimensional standard Brownian motion. Let $(\Omega,\mathcal{F},P_0)$ be a probability space such that $\Omega=C_0(\mathbb{R}_+,\mathbb{R}^2)$, the space of all $\mathbb{R}^2$-valued continuous functions on $\mathbb{R}_+$ starting at $0$, and $P_0$ is the Wiener measure. Then the canonical process, $B=(B_t^1,B_t^2)_t^\top$, is a two-dimensional standard Brownian motion under $P_0$. Moreover, let $\mathbb{F}=(\mathcal{F}_t)_t$ be the filtration generated by $B$, representing the information available to the investor as time evolves. The process $B$ consists of the risk factors that drive the dynamics of the (risky) assets in the market.
\par The market offers a risk-free asset, determined by the short-term interest rate. The (locally) risk-free asset satisfies the price dynamics
\begin{align*}
dP_t^0=P_t^0r_tdt,
\end{align*}
where the process $r=(r_t)_t$ denotes the short-term interest rate. The short-term interest rate is stochastic and evolves as in the (classical) Vasicek model:
\begin{align*}
dr_t=\kappa(\bar{r}-r_t)dt+\sigma_rdB_t^1.
\end{align*}
Thus, the short-term interest rate is a mean-reverting process with a constant volatility. The parameters $\kappa,\bar{r},\sigma_r\in\mathbb{R}$, where $\sigma_r>0$, represent the mean reversion speed and level and the volatility, respectively.
\par In addition to the risk-free asset, the investor can invest in a zero-coupon bond, which is risky. According to the Vasicek model, a zero-coupon bond with maturity $\bar{T}$ is priced such that its price evolves as
\begin{align*}
dP_t^B=P_t^B\Big(\big(r_t+b(\bar{T}-t)\sigma_r\lambda_B\big)dt-b(\bar{T}-t)\sigma_rdB_t^1\Big).
\end{align*}
The function $b:\mathbb{R}\rightarrow\mathbb{R}$, which is defined by
\begin{align*}
b(\bar{T}-t):=\frac{1}{\kappa}\Big(1-e^{-\kappa(\bar{T}-t)}\Big),
\end{align*}
determines the volatility of the bond in terms of the volatility of the short-term interest rate. It is an increasing function; thus, long-term bonds are riskier than short-term bonds. Since the dynamics of the bond are only driven by the first risk factor, the bond is perfectly negatively correlated with the short-term interest rate. The Sharpe ratio of the bond is represented by $\lambda_B\in\mathbb{R}$, which is assumed to be constant in the Vasicek model. This assumption makes the model tractable but is (of course) very unrealistic. Introducing ambiguity to the model in the following helps to overcome this trade-off.
\par Moreover, the investor can invest in the stock market. The stock market is represented by a single risky asset with price dynamics
\begin{align*}
dP_t^S=P_t^S\Big(\big(r_t+\lambda_S\big)dt+\sigma_S\rho dB_t^1+\sigma_S\sqrt{1-\rho^2}dB_t^2\Big).
\end{align*}
The stock is driven by both risk factors--so the stock price is driven by some additional noise, while the first noise term allows for some correlation between the short-term interest rate and the stock. The correlation and the volatility are represented by the parameters $\rho\in(-1,1)$ and $\sigma_S\in\mathbb{R}$, respectively, where $\sigma_S>0$. In contrast to the bond, $\lambda_S\in\mathbb{R}$ represents the expected excess return over the risk-free asset, i.e., the risk premium on the stock. It is also assumed to be constant (as the Sharpe ratio of the bond), which is (again) unrealistic but needed for tractability reasons. \citet{kornkraft2004} offer a critical discussion in this regard. As mentioned above, the succeeding framework with ambiguity offers a good alternative since it is more realistic but still tractable.
\par The investor participates in the market by choosing a dynamic investment strategy, which can be continuously rebalanced over time. An investment strategy is an $\mathbb{F}$-adapted process $\pi=(\pi_t^B,\pi_t^S)_t^\top$, where $\pi_t^B$ and $\pi_t^S$ are the fractions of wealth invested in the bond and the stock, respectively, at time $t$. The remaining wealth is invested in the risk-free asset. For an investment strategy $\pi$, the investor's wealth, given by the process $W^\pi=(W_t^\pi)_t$, evolves as
\begin{align*}
dW_t^\pi=W_t^\pi\Big(\big(r_t+\pi_t^\top\lambda(t)\big)dt+\pi_t^\top\sigma(t)dB_t\Big),
\end{align*}
starting from the investor's initial wealth, denoted by $W_0$, where
\begin{align*}
\lambda(t):=\begin{pmatrix}b(\bar{T}-t)\sigma_r\lambda_B\\\lambda_S\end{pmatrix},\quad\sigma(t):=\begin{pmatrix}-b(\bar{T}-t)\sigma_r&0\\\sigma_S\rho&\sigma_S\sqrt{1-\rho^2}\end{pmatrix}.
\end{align*}
It should be noted that the dynamics of the wealth process ensure that the investment strategy is self-financing. The investor is restricted to choose a sufficiently regular investment strategy, ensuring in turn that the wealth process is sufficiently regular. This is met if the investment strategy $\pi$ is bounded. Such strategies are referred to as admissible investment strategies. The restriction to bounded investment strategies can be relaxed by imposing suitable integrability conditions, but this leads to some technical difficulties when solving the optimal investment problem, which is however still feasible. It turns out that the optimal investment strategy is bounded in any case--thus, the restriction simplifies the exposition but does not restrict the results.
\par Since there is ambiguity, the investor considers several scenarios for the risk premia, the volatilities, and the correlation of the available assets instead of a fixed parameter combination. A fixed parameter combination is given by $(\lambda_B,\lambda_S,\sigma_r,\sigma_S,\rho)$ in the setting above. The space of possible scenarios is denoted by $\Theta$, which consists of all $\mathbb{F}$-adapted processes $\theta=(\lambda_t^B,\lambda_t^S,\sigma_t^r,\sigma_t^S,\rho_t)_t$ satisfying for all $t$
\begin{gather*}
\underline{\lambda}_B(t)\leq\lambda_t^B\leq\overline{\lambda}_B(t),\quad\underline{\lambda}_S\leq\lambda_t^S\leq\overline{\lambda}_S,\quad\underline{\sigma}_r\leq\sigma_t^r\leq\overline{\sigma}_r,\quad\underline{\sigma}_S\leq\sigma_t^S\leq\overline{\sigma}_S,\quad\underline{\rho}\leq\rho_t\leq\overline{\rho},
\end{gather*}
where $\underline{\lambda}_B(t),\underline{\lambda}_S,\underline{\sigma}_r,\underline{\sigma}_S>0$ and $-1<\underline{\rho}\leq0\leq\overline{\rho}<1$. Hence, the investor considers a collection of possible parameter combinations including time varying risk premia, volatilities, and correlations. The only restriction is that their evolution is bounded by some extreme values. The extreme values for the risk premium on the stock, both volatilities, and the correlation are given exogenously. The extreme values for the risk premium on the bond are determined endogenously by the volatility of the short-term interest rate. The reason is that an arbitrage-free term structure in the presence of ambiguous volatility requires an ambiguous process in the risk premium of a zero-coupon bond \citep{holzermann2021,holzermann2022}. The same can be deduced for fixed coupon bonds and other interest rate derivatives \citep{holzermann2022'}. In a Vasicek model with ambiguous volatility, the risk premium on a zero-coupon bond with maturity $\bar{T}$ is $b(\bar{T}-t)\lambda_t^B$, where the ambiguity in the risk premium, represented by $\lambda_t^B$, is bounded by
\begin{align*}
\underline{\lambda}_B(t)&=e^{-2\kappa t}\lambda_0^B+\frac{\underline{\sigma}_r^2}{2\kappa}\Big(1-e^{-2\kappa t}\Big),
\\\overline{\lambda}_B(t)&=e^{-2\kappa t}\lambda_0^B+\frac{\overline{\sigma}_r^2}{2\kappa}\Big(1-e^{-2\kappa t}\Big)
\end{align*}
\citep[Section 6]{holzermann2021}. The extreme values for the ambiguous risk premium on the bond are essentially determined by the extreme values for the ambiguous variance of the short-term interest rate \citep[Theorem 2.1]{holzermann2021}. Compared to the original results, $\lambda_0^B$ is nonzero, since otherwise, the initial value of the risk premium on the bond is zero. This still yields an arbitrage-free description of the bond price, since the dynamics of the ambiguous part in the risk premium on the bond are essential for obtaining an arbitrage-free term structure \citep[Proof of Theorem 6.1]{holzermann2021}.
\par Consequently, the investor's wealth and the short-term interest rate have a different evolution in each possible scenario. For an investment strategy $\pi$ and a scenario $\theta\in\Theta$, the investor's wealth and the short-term interest rate, now denoted by $W^{\pi,\theta}=(W_t^{\pi,\theta})_t$ and $r^\theta=(r_t^\theta)_t$, respectively, evolve as
\begin{align*}
dW_t^{\pi,\theta}&=W_t^{\pi,\theta}\Big(\big(r_t^\theta+\pi_t^\top\lambda_t^\theta\big)dt+\pi_t^\top\sigma_t^\theta dB_t\Big),
\\dr_t^\theta&=\kappa(\bar{r}-r_t)dt+\nu_t^\theta dB_t,
\end{align*}
starting from the investor's initial wealth and the initial short-term interest rate, denoted by $W_0$ and $r_0$, respectively, where
\begin{align*}
\lambda_t^\theta:=\begin{pmatrix}b(\bar{T}-t)\lambda_t^B\\\lambda_t^S\end{pmatrix},\quad\sigma_t^\theta:=\begin{pmatrix}-b(\bar{T}-t)\sigma_t^r&0\\\sigma_t^S\rho_t&\sigma_t^S\sqrt{1-\rho_t^2}\end{pmatrix},\quad\nu_t^\theta:=\begin{pmatrix}\sigma_t^r\\0\end{pmatrix}.
\end{align*}
The investor is still restricted to choose a sufficiently regular investment strategy. The space of admissible investment strategies is denoted by $\Pi$, which consists of all $\mathbb{F}$-adapted processes $\pi$ that are bounded, ensuring that the wealth process is sufficiently regular for all $\theta\in\Theta$, that is, in each possible scenario.
\par The ambiguity is represented by a set of priors, where each prior represents a different scenario. Similar to \citet{epsteinji2013,epsteinji2014}, one can construct a prior related to each scenario for the evolution of the wealth process and the short-term interest rate. For each investment strategy $\pi\in\Pi$ and each possible scenario $\theta\in\Theta$, define the prior
\begin{align*}
P_{\pi,\theta}:=P_0\circ(W^{\pi,\theta},r^\theta)^{-1},
\end{align*}
i.e., the probability measure induced by the processes $W^{\pi,\theta}$ and $r^\theta$. Then the set of priors, related to the investment strategy $\pi\in\Pi$, is defined by
\begin{align*}
\mathcal{P}_\pi:=\{P_{\pi,\theta}\,\vert\,\theta\in\Theta\}.
\end{align*}
Hence, the canonical process, now denoted by $(W,r)=(W_t,r_t)_t$, which represents the wealth and the short-term interest rate, evolves according to a different scenario under each prior.
\par The investor has preferences about terminal wealth from investing, displaying both risk and ambiguity aversion. As in similar studies, the investor cares about terminal wealth. The preferences are represented by expected utility with a constant relative risk aversion (CRRA) utility function, representing the investor's risk aversion. The attitude towards ambiguity is incorporated by maxmin expected utility in the spirit of \citet{gilboaschmeidler1989}: the investor ranks investment strategies according to their expected utility from terminal wealth under the worst-case scenario among the relevant priors. Thus, in order to rank an admissible investment strategy $\pi\in\Pi$, the investor considers
\begin{align*}
\inf_{P\in\mathcal{P}_\pi}\mathbb{E}_P[u(W_T)],
\end{align*}
where $T<\infty$ represents the investor's time horizon and $u$ is a CRRA utility function: $u(W):=\frac{W^{1-\gamma}}{1-\gamma}$ for a relative risk aversion parameter $\gamma>1$. The case $\gamma=1$ corresponds to logarithmic utility: $u(W)=\ln(W)$.
\par Then the aim of the investor is to choose an investment strategy maximizing the preferences, that is, an optimal investment strategy in the following sense.
\begin{df}\label{optimal investment strategy}
An admissible investment strategy $\hat{\pi}\in\Pi$ is called optimal investment strategy if it holds
\begin{align*}
\sup_{\pi\in\Pi}\inf_{P\in\mathcal{P}_\pi}\mathbb{E}_P[u(W_T)]=\inf_{P\in\mathcal{P}_{\hat{\pi}}}\mathbb{E}_P[u(W_T)].
\end{align*}
\end{df}
\noindent Apart from the optimal investment strategy, it is also interesting to determine the worst-case prior related to the optimal investment strategy of the investor.
\begin{df}\label{worst-case prior}
A prior $\hat{P}\in\mathcal{P}_{\hat{\pi}}$ is called worst-case prior for an admissible investment strategy $\hat{\pi}\in\Pi$ if it holds
\begin{align*}
\inf_{P\in\mathcal{P}_{\hat{\pi}}}\mathbb{E}_P[u(W_T)]=\mathbb{E}_{\hat{P}}[u(W_T)].
\end{align*}
\end{df}

\section{Optimal Investment}\label{optimal investment}
The optimal investment problem admits a closed-form solution for the optimal investment strategy, the worst-case prior, and the related value function.
\begin{thm}\label{solution to the optimal investment problem}
The optimal investment strategy is given by $\hat{\pi}=(\hat{\pi}_t^B,\hat{\pi}_t^S)_t^\top$, where
\begin{align*}
\hat{\pi}_t^B&:=\frac{1}{\gamma}\frac{1}{1-\hat{\rho}(t)^2}\frac{1}{b(\bar{T}-t)\hat{\sigma}_r}\biggl(\frac{\hat{\lambda}_B(t)}{\hat{\sigma}_r}+\hat{\rho}(t)\frac{\hat{\lambda}_S}{\hat{\sigma}_S}\biggr)+\frac{\gamma-1}{\gamma}\frac{b(T-t)}{b(\bar{T}-t)},
\\\hat{\pi}_t^S&:=\frac{1}{\gamma}\frac{1}{1-\hat{\rho}(t)^2}\frac{1}{\hat{\sigma}_S}\biggl(\frac{\hat{\lambda}_S}{\hat{\sigma}_S}+\hat{\rho}(t)\frac{\hat{\lambda}_B(t)}{\hat{\sigma}_r}\biggr),
\end{align*}
and the worst-case prior for $\hat{\pi}$ is given by $P_{\hat{\pi},\hat{\theta}}$, where $\hat{\theta}=(\hat{\lambda}_B(t),\hat{\lambda}_S,\hat{\sigma}_r,\hat{\sigma}_S,\hat{\rho}(t))_t$ and
\begin{gather*}
\hat{\lambda}_B(t):=\underline{\lambda}_B(t),\quad\hat{\lambda}_S:=\underline{\lambda}_S,\quad\hat{\sigma}_r:=\overline{\sigma}_r,\quad\hat{\sigma}_S:=\overline{\sigma}_S,
\\\hat{\rho}(t):=\max\biggl\{\underline{\rho},-\frac{\underline{\lambda}_B(t)/\overline{\sigma}_r}{\underline{\lambda}_S/\overline{\sigma}_S},-\frac{\underline{\lambda}_S/\overline{\sigma}_S}{\underline{\lambda}_B(t)/\overline{\sigma}_r}\biggr\}.
\end{gather*}
The value of the optimal investment problem is given by
\begin{align*}
\sup_{\pi\in\Pi}\inf_{P\in\mathcal{P}_\pi}\mathbb{E}_P[u(W_T)]=V(0,W_0,r_0),
\end{align*}
where the value function $V:[0,T]\times\mathbb{R}_+\times\mathbb{R}\rightarrow\mathbb{R}$ is defined by
\begin{align*}
V(t,W,r):=\exp\Big((1-\gamma)\big(a_0(t)+a_1(t)r\big)\Big)\frac{W^{1-\gamma}}{1-\gamma}
\end{align*}
and the functions $a_0,a_1:[0,T]\rightarrow\mathbb{R}$ are defined by
\begin{align*}
a_0(t):={}&\frac{1}{2}\frac{1}{\gamma}\int_t^T\frac{1}{1-\hat{\rho}(u)^2}\biggl(\frac{\hat{\lambda}_B(u)^2}{\hat{\sigma}_r^2}+2\hat{\rho}(u)\frac{\hat{\lambda}_B(u)\hat{\lambda}_S}{\hat{\sigma}_r\hat{\sigma}_S}+\frac{\hat{\lambda}_S^2}{\hat{\sigma}_S^2}\biggr)du+\kappa\bar{r}\int_t^Tb(T-u)du
\\&+\frac{\gamma-1}{\gamma}\int_t^T\hat{\lambda}_B(u)b(T-u)du-\frac{1}{2}\frac{\gamma-1}{\gamma}\hat{\sigma}_r^2\int_t^Tb(T-u)^2du
\end{align*}
and $a_1(t):=b(T-t)$, respectively.
\end{thm}
\begin{rem}
The case with logarithmic utility yields the same strategies and worst-case prior but a different value function. One can check (as in Section \ref{proof}) that setting $\gamma=1$ in the expression for $\hat{\pi}$ in Theorem \ref{solution to the optimal investment problem} yields the optimal investment strategy with the same worst-case prior for logarithmic utility. Then the additional term in $\hat{\pi}^B$ vanishes, which represents the fact that investors with logarithmic utility do not hedge interest rate risk (see the discussion below). However, similar to the case without ambiguity, the value function has a different structure when the utility function is of logarithmic type.
\end{rem}
\par The solution to the optimal investment problem shows that the ambiguity affects only the speculative motives of the investor but not the demand for hedging interest rate risk. The optimal investment strategy consists of two parts--a speculative part and a part that hedges interest rate risk--as in the case without ambiguity. The optimal portfolio, $\hat{\pi}$, from Theorem \ref{solution to the optimal investment problem} can be separated into two parts:
\begin{align*}
\hat{\pi}=\frac{1}{\gamma}\hat{\pi}^\text{myopic}+\frac{\gamma-1}{\gamma}\pi^\text{hedge},
\end{align*}
where the investor's level of risk aversion, $\gamma$, determines how wealth is allocated between both parts. The first part, $\hat{\pi}^\text{myopic}=(\hat{\pi}_t^\text{myopic})_t$, is given by
\begin{align*}
\hat{\pi}_t^\text{myopic}:=\frac{1}{1-\hat{\rho}(t)^2}\begin{pmatrix}\frac{1}{b(\bar{T}-t)\hat{\sigma}_r}\Big(\frac{\hat{\lambda}_B(t)}{\hat{\sigma}_r}+\hat{\rho}(t)\frac{\hat{\lambda}_S}{\hat{\sigma}_S}\Big)\\\frac{1}{\hat{\sigma}_S}\Big(\frac{\hat{\lambda}_S}{\hat{\sigma}_S}+\hat{\rho}(t)\frac{\hat{\lambda}_B(t)}{\hat{\sigma}_r}\Big)\end{pmatrix}
\end{align*}
and represents the speculative part, since it corresponds to the classical mean-variance optimal investment strategy, which is typically referred to as the myopic portfolio. The second part, $\pi^\text{hedge}=(\pi_t^\text{hedge})_t$, is given by
\begin{align*}
\pi_t^\text{hedge}:=\begin{pmatrix}\frac{b(T-t)}{b(\bar{T}-t)}\\0\end{pmatrix},
\end{align*}
which is referred to as the hedge part, since it can be related to hedging interest rate risk. As in the classical situation without ambiguity, only the bond is used to hedge interest rate risk. In particular, one can see that the ambiguity, represented by the worst-case scenario, $\hat{\theta}$, only affects the speculative part of the investment, but it has no effect on hedging interest rate risk. This is due to the fact that the source of ambiguity in the volatility of the short-term interest rate and the bond is the same and thus cancels out.
\par The effect of the ambiguity on the speculative part of the investment strategy can also be interpreted as a hedge against the ambiguity. Similar to \citet{florlarsen2014} and \citet{maenhout2006}, one can alternatively represent the optimal portfolio by three parts, where one measures the difference to the optimal portfolio neglecting ambiguity. If there is no ambiguity, that is, if the investor only considers the (constant) scenario $\tilde{\theta}=(\lambda_0^B,\lambda_S,\sigma_r,\sigma_S,\rho)_t$ (similar to the beginning of Section \ref{investment problem}), the optimal portfolio, denoted by $\tilde{\pi}=(\tilde{\pi}_t)_t$, is given by
\begin{align*}
\tilde{\pi}=\frac{1}{\gamma}\tilde{\pi}^\text{myopic}+\frac{\gamma-1}{\gamma}\pi^\text{hedge},
\end{align*}
where in this case, the speculative part, $\tilde{\pi}^\text{myopic}=(\tilde{\pi}_t^\text{myopic})_t$, is instead determined by the scenario $\tilde{\theta}$:
\begin{align*}
\tilde{\pi}_t^\text{myopic}:=\frac{1}{1-\rho^2}\begin{pmatrix}\frac{1}{b(\bar{T}-t)\sigma_r}\Big(\frac{\lambda_0^B}{\sigma_r}+\rho\frac{\lambda_S}{\sigma_S}\Big)\\\frac{1}{\sigma_S}\Big(\frac{\lambda_S}{\sigma_S}+\rho\frac{\lambda_0^B}{\sigma_r}\Big)\end{pmatrix},
\end{align*}
and the hedge part, $\pi^\text{hedge}$, is the same as before. The optimal investment strategy with ambiguity, $\hat{\pi}$, can then be written as
\begin{align*}
\hat{\pi}=\frac{1}{\gamma}\tilde{\pi}^\text{myopic}+\frac{\gamma-1}{\gamma}\pi^\text{hedge}+\frac{1}{\gamma}(\hat{\pi}^\text{myopic}-\tilde{\pi}^\text{myopic}),
\end{align*}
i.e., as the optimal investment strategy neglecting ambiguity, consisting of the two parts $\tilde{\pi}^\text{myopic}$ and $\pi^\text{hedge}$, and an additional part. The additional part, $\hat{\pi}^\text{myopic}-\tilde{\pi}^\text{myopic}$, can be interpreted as a hedge against ambiguity. It is determined by how much the speculative part of the investment strategy is affected by the ambiguity (compared to the reference scenario $\tilde{\theta}$).
\par Whether the investor invests in the bond or the stock for speculative reasons is determined by the worst-case scenario for the correlation. Theorem \ref{solution to the optimal investment problem} shows that the worst-case scenario for the correlation depends on the lowest possible correlation, $\underline{\rho}$, and the worst-case Sharpe ratios of the bond and the stock, $\underline{\lambda}_B(t)/\overline{\sigma}_r$ and $\underline{\lambda}_S/\overline{\sigma}_S$, respectively. There are three possible cases at each time $t$:
\begin{enumerate}
\item[$(i)$] The lowest possible correlation is relatively high and/or the worst-case Sharpe ratios are relatively close to each other, that is,
\begin{align*}
\underline{\rho}>-\frac{\underline{\lambda}_B(t)/\overline{\sigma}_r}{\underline{\lambda}_S/\overline{\sigma}_S},\quad\underline{\rho}>-\frac{\underline{\lambda}_S/\overline{\sigma}_S}{\underline{\lambda}_B(t)/\overline{\sigma}_r}.
\end{align*}
Then the portfolio weights in the speculative part, $\hat{\pi}_t^\text{myopic}$, are both positive, since it holds
\begin{gather*}
\frac{\hat{\lambda}_B(t)}{\hat{\sigma}_r}+\hat{\rho}(t)\frac{\hat{\lambda}_S}{\hat{\sigma}_S}=\frac{\underline{\lambda}_B(t)}{\overline{\sigma}_r}+\underline{\rho}\frac{\underline{\lambda}_S}{\overline{\sigma}_S}>0,
\\\frac{\hat{\lambda}_S}{\hat{\sigma}_S}+\hat{\rho}(t)\frac{\hat{\lambda}_B(t)}{\hat{\sigma}_r}=\frac{\underline{\lambda}_S}{\overline{\sigma}_S}+\underline{\rho}\frac{\underline{\lambda}_B(t)}{\overline{\sigma}_r}>0.
\end{gather*}
Consequently, the investor invests in the bond and in the stock for speculative reasons.
\item[$(ii)$] The lowest possible correlation is relatively low and/or the worst-case Sharpe ratio of the bond is much lower than the worst-case Sharpe ratio of the stock, that is,
\begin{align*}
-\frac{\underline{\lambda}_B(t)/\overline{\sigma}_r}{\underline{\lambda}_S/\overline{\sigma}_S}\geq\underline{\rho}.
\end{align*}
Since $\underline{\rho}>-1$ and consequently $\underline{\lambda}_B(t)/\overline{\sigma}_r<\underline{\lambda}_S/\overline{\sigma}_S$, it then holds
\begin{gather*}
\frac{\hat{\lambda}_B(t)}{\hat{\sigma}_r}+\hat{\rho}(t)\frac{\hat{\lambda}_S}{\hat{\sigma}_S}=\frac{\underline{\lambda}_B(t)}{\overline{\sigma}_r}-\frac{\underline{\lambda}_B(t)/\overline{\sigma}_r}{\underline{\lambda}_S/\overline{\sigma}_S}\frac{\underline{\lambda}_S}{\overline{\sigma}_S}=0,
\\\frac{\hat{\lambda}_S}{\hat{\sigma}_S}+\hat{\rho}(t)\frac{\hat{\lambda}_B(t)}{\hat{\sigma}_r}=\frac{\underline{\lambda}_S}{\overline{\sigma}_S}-\frac{\underline{\lambda}_B(t)/\overline{\sigma}_r}{\underline{\lambda}_S/\overline{\sigma}_S}\frac{\underline{\lambda}_B(t)}{\overline{\sigma}_r}>0.
\end{gather*}
Thus, the investor does not invest in the bond for speculative reasons but only in order to hedge interest rate risk--the investor always hedges interest rate risk, since the hedge part, $\pi^\text{hedge}$, is not affected by the ambiguity.
\item[$(iii)$] The lowest possible correlation is relatively low and/or the worst-case Sharpe ratio of the stock is much lower than the worst-case Sharpe ratio of the bond, that is,
\begin{align*}
-\frac{\underline{\lambda}_S/\overline{\sigma}_S}{\underline{\lambda}_B(t)/\overline{\sigma}_r}\geq\underline{\rho}.
\end{align*}
Then, similar to the previous case, it holds
\begin{gather*}
\frac{\hat{\lambda}_B(t)}{\hat{\sigma}_r}+\hat{\rho}(t)\frac{\hat{\lambda}_S}{\hat{\sigma}_S}=\frac{\underline{\lambda}_B(t)}{\overline{\sigma}_r}-\frac{\underline{\lambda}_S/\overline{\sigma}_S}{\underline{\lambda}_B(t)/\overline{\sigma}_r}\frac{\underline{\lambda}_S}{\overline{\sigma}_S}>0,
\\\frac{\hat{\lambda}_S}{\hat{\sigma}_S}+\hat{\rho}(t)\frac{\hat{\lambda}_B(t)}{\hat{\sigma}_r}=\frac{\underline{\lambda}_S}{\overline{\sigma}_S}-\frac{\underline{\lambda}_S/\overline{\sigma}_S}{\underline{\lambda}_B(t)/\overline{\sigma}_r}\frac{\underline{\lambda}_B(t)}{\overline{\sigma}_r}=0.
\end{gather*}
Therefore, the investor invests in the bond for speculative reasons and in order to hedge interest rate risk but refrains from investing in the stock--the optimal portfolio in this case is a sparse portfolio with no stock investment. However, such a portfolio is by no means extreme, since this case leads to a short position in the stock if there is no ambiguity.
\end{enumerate}
\par The worst-case scenario also shows that the endogenous ambiguity leads to an additional dynamic effect in the optimal investment strategy compared to the case without ambiguity. Without ambiguity, the dynamic behavior of the optimal investment strategy, $\tilde{\pi}$, is represented by the function $b$ which depends on time. More specifically, it depends on the time to the investor's time horizon and the time to the bond's maturity, $T-t$ and $\bar{T}-t$, respectively. Thus, there is no difference between the evolution of time and shortening the time horizon of the investor and the maturity of the bond. The effect prevails in the case with ambiguity, but there is also an additional dynamic effect. The endogenous ambiguity about the risk premium on the bond is time-dependent and enters the optimal investment strategy. It also makes the worst-case correlation time-dependent, which enters the optimal investment strategy. Compared to the dynamic effect from the case without ambiguity, this additional dynamic effect is independent of the investor's time horizon and the bond's maturity. How this changes the behavior of the optimal investment strategy over time is shown by the implementation of the optimal investment strategy in the following section.

\section{Implementation}\label{implementation}
The calibration of the model is similar to the case without ambiguity and incorporates ambiguity by inferring the extreme values for the ambiguous quantities from the extreme values of rolling window estimates. The investor's time horizon is $T=10$ and the maturity of the bond is $\bar{T}=20$, which is a similar choice compared to other papers on optimal investment with stochastic interest rates \citep{bajeux-besnainouetal2001,brennanxia2000,florlarsen2014,sorensen1999}. If there is no ambiguity, i.e., when the investor only considers the (constant) scenario $\tilde{\theta}$ (introduced in the previous section), the parameters of the short-term interest rate, the bond, and the stock can be estimated as outlined in Section \ref{calibration negelcting ambiguity} of the appendix. Using US data ranging from January, 1946, to December, 2020, the estimation then yields $\kappa=0.336$, $\bar{r}=0.0381$, $\sigma_r=0.0262$, $\lambda_0^B=0.0086$, $\lambda_S=0.078$, $\sigma_S=0.1457$, and $\rho=0.0196$, which determine the optimal investment strategy neglecting ambiguity. The estimation taking ambiguity into account is based on the same approach but uses rolling window estimates (with a rolling window of 20 years) for the ambiguous quantities (to observe how they vary over time) and then takes the minimum and the maximum of the rolling window estimates as the extreme values for the ambiguous quantities. The variation of the rolling window estimates over time can be seen in Figure \ref{rolling estimates} and leads to the estimates $\underline{\sigma}_r=0.0111$, $\overline{\sigma}_r=0.0455$, $\underline{\lambda}_S=0.0124$, $\overline{\lambda}_S=0.1254$, $\underline{\sigma}_S=0.1191$, $\overline{\sigma}_S=0.1651$, $\underline{\rho}=-0.1474$, and $\overline{\rho}=0.1337$.
\begin{figure}
\includegraphics[width=\textwidth]{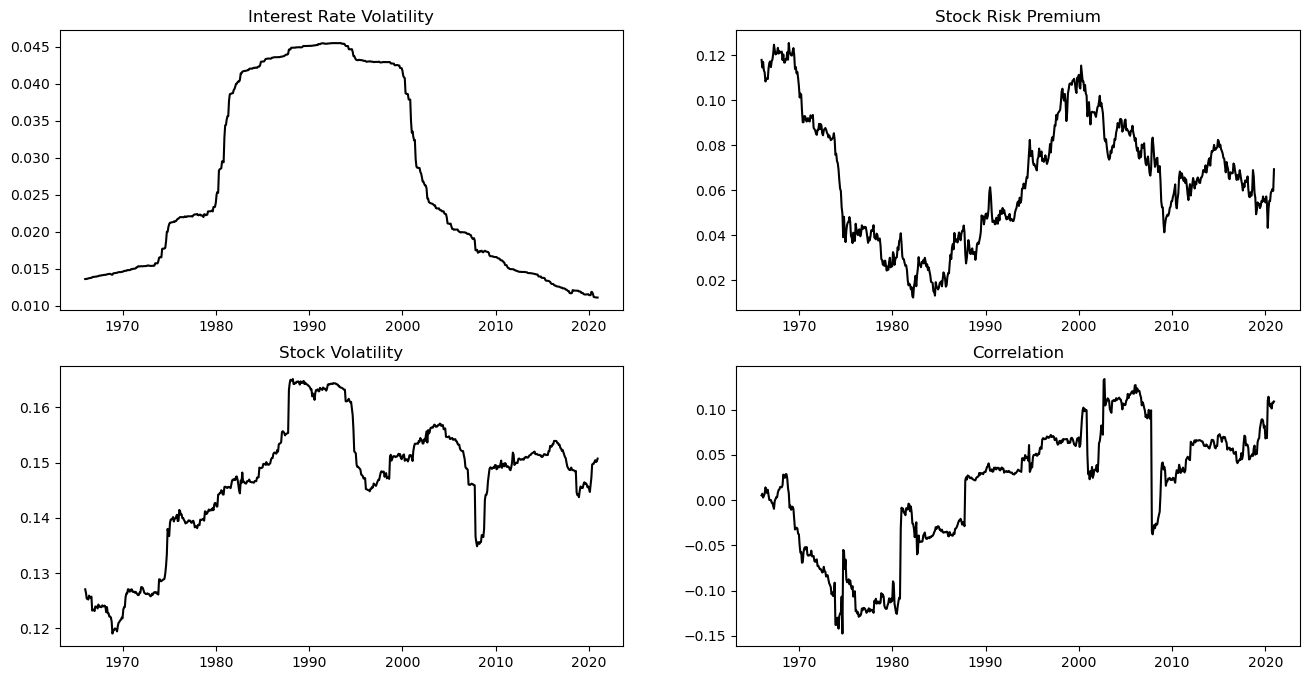}
\caption{Rolling window estimates for the volatility of the short-term interest rate, the risk premium on the stock, the volatility of the stock, and the correlation between the short-term interest rate and the stock.}\label{rolling estimates}
\end{figure}
\par The estimates then enable to compute the optimal investment strategies of an ambiguity averse investor and an investor neglecting ambiguity, respectively. The optimal investment strategy neglecting ambiguity, $\tilde{\pi}$, can be computed using the estimates from the case without ambiguity. In the case with ambiguity, the estimates for the extreme values of the ambiguous quantities determine the worst-case scenario, $\hat{\theta}$, as shown in Theorem \ref{solution to the optimal investment problem}. In particular, the lowest possible correlation, $\underline{\rho}$, and the ratios $-(\underline{\lambda}_B(t)/\overline{\sigma}_r)/(\underline{\lambda}_S/\overline{\sigma}_S)$ and $-(\underline{\lambda}_S/\overline{\sigma}_S)/(\underline{\lambda}_B(t)/\overline{\sigma}_r)$ determine the worst-case scenario for the correlation, which is illustrated by Figure \ref{worst-case correlation} and affects the speculative motives of the investor (as described in the cases $(i)$, $(ii)$, and $(iii)$ in the previous section). From the worst-case scenario, one can then compute the optimal investment strategy with ambiguity aversion, $\hat{\pi}$.
\begin{figure}
\includegraphics[width=\textwidth]{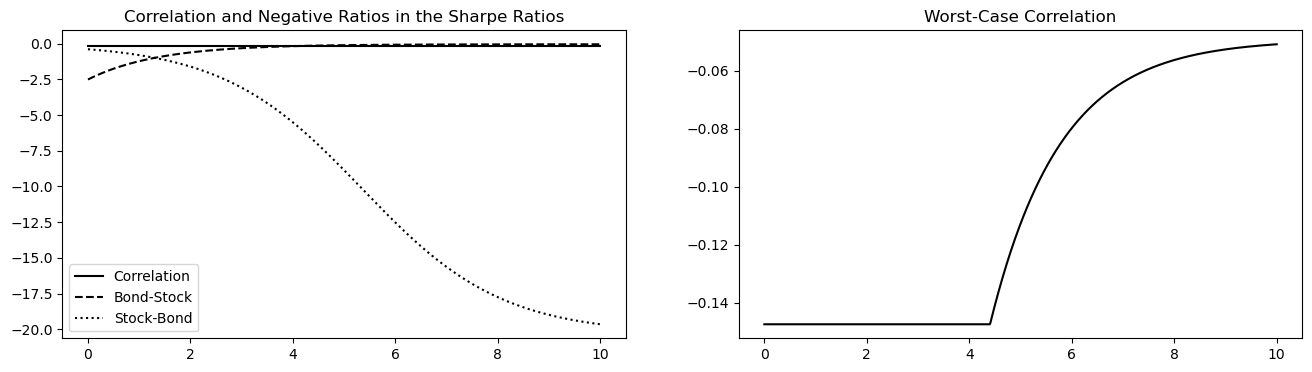}
\caption{The lowest possible correlation, $\underline{\rho}$, the negative ratio in the worst-case Sharpe ratios of the bond and the stock, $\underline{\lambda}_B(t)/\overline{\sigma}_r$ and $\underline{\lambda}_S/\overline{\sigma}_S$, respectively, and its reciprocal (on the left) and the worst-case scenario for the correlation (on the right).}\label{worst-case correlation}
\end{figure}
\par A comparison of the optimal investment strategies at inception shows that ambiguity aversion yields more reasonable positions and preserves the desirable properties of the optimal investment strategy neglecting ambiguity. Figure \ref{weights depending on risk aversion} shows the portfolio weights at time $0$ and the related bond-stock ratio as a function of the investor's level of risk aversion, $\gamma$, corresponding to the optimal investment strategy with ambiguity aversion and the optimal investment strategy neglecting ambiguity. It shows that the ambiguity decreases the bond and the stock weight and increases the cash weight. Thus, ambiguity aversion leads to less leveraged positions compared to the very extreme positions neglecting ambiguity. This is due to the hedge against the ambiguity, which is the distance between the two curves of the risky assets' portfolio weights at $\gamma=1$. The difference in the portfolio weights is greater for low levels of risk aversion, since then there is more investment in the hedge against ambiguity. Since the hedge part in the bond weight is not affected by the ambiguity, the ambiguity generally increases the bond-stock ratio. The effect is more pronounced for high levels of risk aversion. Therefore, the ambiguity amplifies the effect that the bond-stock ratio is increasing in the relative risk aversion, which is an important result in the asset allocation literature (as pointed out in the introduction) and hence robust with respect to ambiguity.
\begin{figure}
\includegraphics[width=\textwidth]{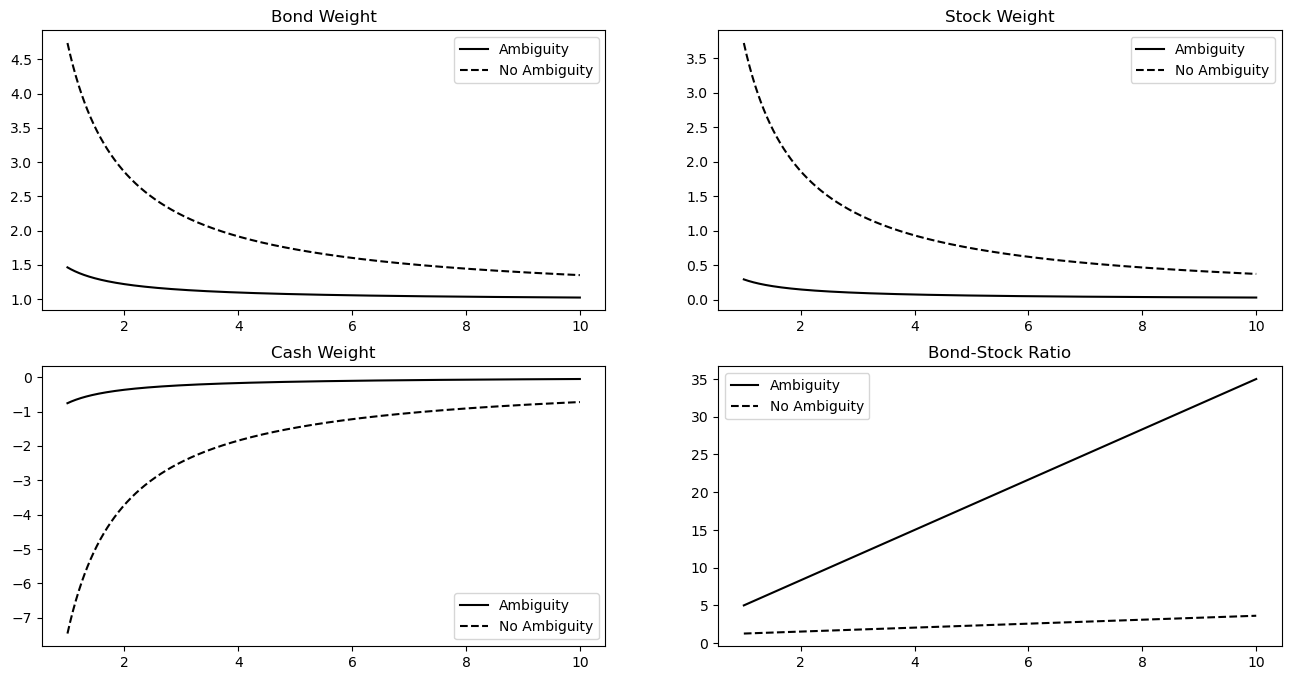}
\caption{The optimal (static) portfolio weights and the bond-stock ratio depending on the relative risk aversion for an ambiguity averse investor and an investor neglecting ambiguity.}\label{weights depending on risk aversion}
\end{figure}
\par While all sources of ambiguity affect the optimal investment strategy, the magnitude of their effects is different. Figure \ref{weights depending on risk aversion with ambiguity} shows the same as Figure \ref{weights depending on risk aversion} but for the optimal investment strategy considering only one source of ambiguity for each of the different sources of ambiguity. It shows that the ambiguity about the volatility of the short-term interest rate and the risk premium on the stock have the largest impact on the optimal investment strategy. The ambiguous volatility of the short-term interest rate decreases the bond weight the most compared to the other sources of ambiguity. The ambiguous risk premium on the stock on the other hand has the largest decreasing effect on the stock weight compared to the other sources of ambiguity and consequently has the largest increasing effect on the bond-stock ratio. The reason is that, according to the estimates, the ambiguity about the volatility of the short-term interest rate is relatively large (i.e., the distance between the extreme values relative to the values is large) compared to the ambiguity about the volatility of the stock and that there is more ambiguity (i.e., a larger distance between the extreme values) about the stock's risk premium compared to its volatility. Despite its magnitude, the ambiguity about the correlation has a relatively modest effect on the optimal investment strategy.
\begin{figure}
\includegraphics[width=\textwidth]{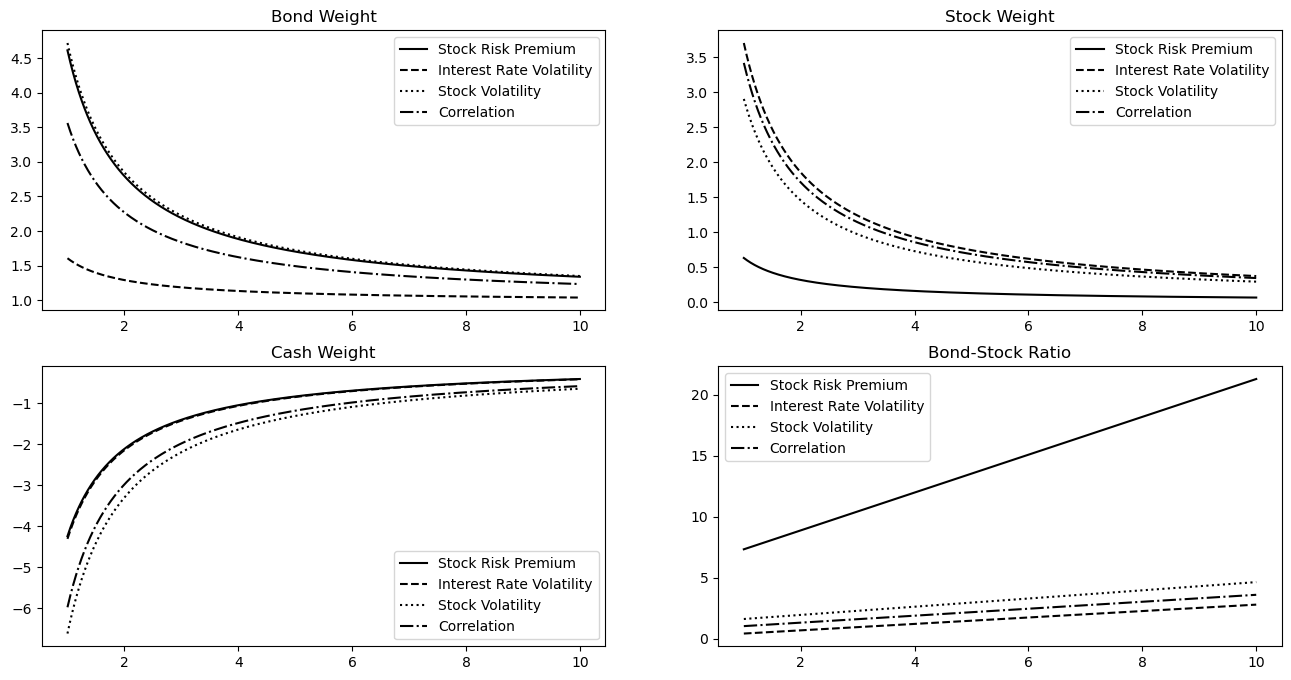}
\caption{The optimal (static) portfolio weights and the bond-stock ratio depending on the relative risk aversion for an ambiguity averse investor considering only one source of ambiguity for each source of ambiguity separately.}\label{weights depending on risk aversion with ambiguity}
\end{figure}
\par In addition to the static effects, ambiguity aversion also leads to a much more natural dynamic behavior. Figure \ref{weights as time evolves} shows the portfolio weights, the speculative part and the hedge part in the bond weight, and the bond-stock ratio as a function of time for $\gamma=2$ corresponding to the optimal investment strategy with ambiguity aversion and the optimal investment strategy neglecting ambiguity. The optimal investment strategy neglecting ambiguity changes only slightly over time (due to the hedge part in the bond weight, which decreases the bond weight and increases the cash weight slightly close to the time horizon) and therefore stays highly leveraged as time passes. The stock weight in the optimal investment strategy with ambiguity aversion increases slightly, as the speculative part shifts from the bond to the stock over time, and then stays constant. The bond weight on the other hand decreases vastly as time passes, since the speculative part decreases (in addition to the hedge part), until the investor stops investing in the bond for speculative reasons, as described in the case $(ii)$ in Section \ref{optimal investment}, which is illustrated by Figure \ref{worst-case correlation}, and invests in the bond only in order to hedge interest rate risk. Hence, the cash weight increases over time such that the investor holds almost all wealth in cash at the time horizon, which is intuitive and resembles popular investment advice in contrast to the optimal investment strategy neglecting ambiguity.
\begin{figure}
\includegraphics[width=\textwidth]{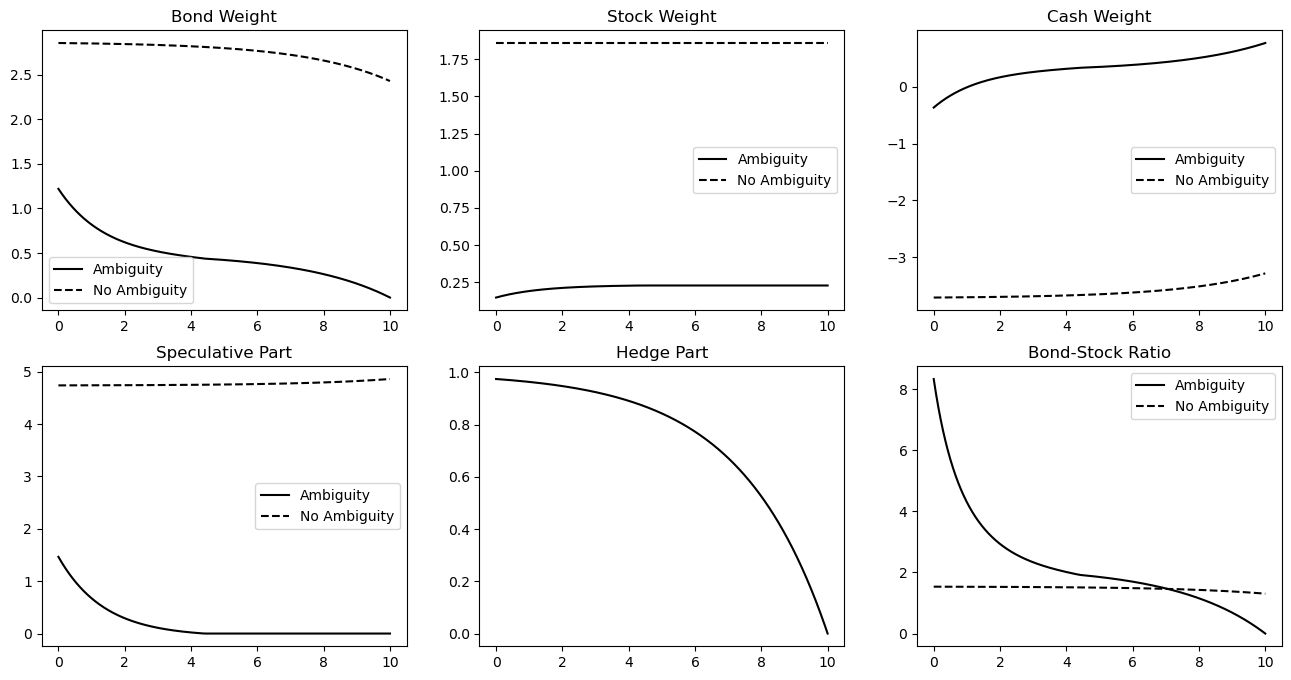}
\caption{The optimal portfolio weights, the speculative part and the hedge part in the bond weight, and the bond-stock ratio depending on time for $\gamma=2$ for an ambiguity averse investor and an investor neglecting ambiguity.}\label{weights as time evolves}
\end{figure}
\par Comparing the different sources of ambiguity shows which source creates the difference in the dynamic behavior of the optimal investment strategy with ambiguity aversion. Figure \ref{weights as time evolves with ambiguity} shows the same as Figure \ref{weights as time evolves} but for the optimal investment strategy considering only one source of ambiguity for each of the different sources of ambiguity (similar to Figure \ref{weights depending on risk aversion with ambiguity} for the static portfolio weights). It shows that the ambiguous volatility of the short-term interest rate is responsible for the change in the dynamic behavior of the optimal investment strategy as it is outlined above. The reason is that the endogenous ambiguity, which is due to the ambiguity about the volatility of the short-term interest rate, causes an additional dynamic effect in the optimal investment strategy, as described at the end of Section \ref{optimal investment}. The other sources of ambiguity in contrast do not change the dynamic behavior; instead, they make the optimal investment strategy less extreme (similar to the static effects seen in Figure \ref{weights depending on risk aversion with ambiguity}).
\begin{figure}
\includegraphics[width=\textwidth]{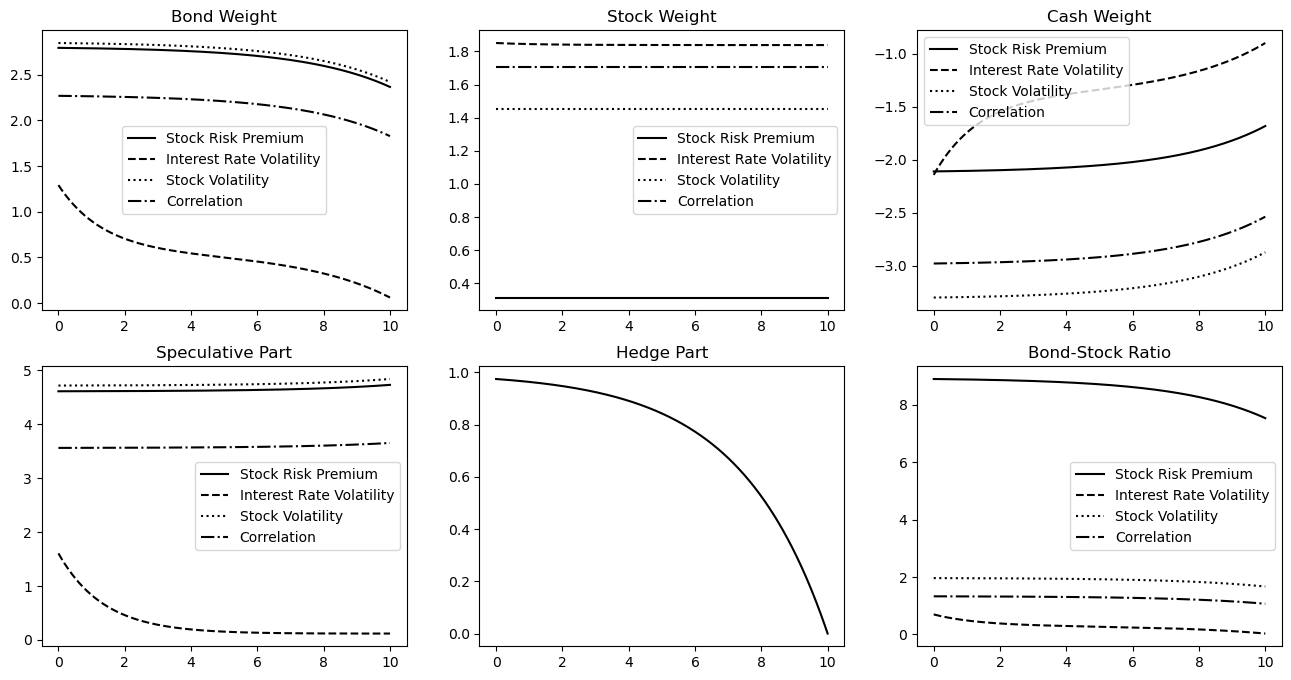}
\caption{The optimal portfolio weights, the speculative part and the hedge part in the bond weight, and the bond-stock ratio depending on time for $\gamma=2$ for an ambiguity averse investor considering only one source of ambiguity for each source of ambiguity separately.}\label{weights as time evolves with ambiguity}
\end{figure}

\section{Conclusion}\label{conclusion}
This paper shows the importance and advantages of considering ambiguity--especially several sources of ambiguity--in dynamic asset allocation problems with bonds, stocks, and cash as the main asset classes. On a theoretical level, the results of the paper show that ambiguity aversion in the sense of maxmin expected utility only affects the speculative motives of the investor, which can be interpreted as a hedge against the ambiguity, but it does not affect the demand for hedging interest rate risk. From a practical perspective, ambiguity aversion leads to investment strategies that are more in line with popular investment advice due to the dampening effect of ambiguity in the speculative and highly leveraged part of the investment strategy.

\appendix

\section*{Appendix}

\section{Proof of Theorem \ref{solution to the optimal investment problem}}\label{proof}
The approach to solving the optimal investment problem is based on the following extension of the martingale optimality principle to a multiple prior setting.
\begin{prp}\label{martingale optimality principle for optimal strategies}
If there exists a function $V\in C^{1,2,2}([0,T)\times\mathbb{R}_+\times\mathbb{R})$, a strategy $\hat{\pi}\in\Pi$, and a scenario $\hat{\theta}\in\Theta$ such that
\begin{enumerate}
\item[$(i)$] it holds $V(T,W,r)=u(W)$ for all $(W,r)\in\mathbb{R}_+\times\mathbb{R}$,
\item[$(ii)$] the process $(V(t,W_t,r_t))_t$ is a $P_{\hat{\pi},\hat{\theta}}$-martingale,
\item[$(iii)$] for each strategy $\pi\in\Pi$, the process $(V(t,W_t,r_t))_t$ is a $P_{\pi,\hat{\theta}}$-supermartingale,
\item[$(iv)$] the prior $P_{\hat{\pi},\hat{\theta}}$ is a worst-case prior for the strategy $\hat{\pi}$,
\end{enumerate}
then $\hat{\pi}$ is an optimal investment strategy and
\begin{align*}
\sup_{\pi\in\Pi}\inf_{P\in\mathcal{P}_\pi}\mathbb{E}_P[u(W_T)]=V(0,W_0,r_0).
\end{align*}
\end{prp}
\begin{proof}
Having a worst-case prior for the strategy $\hat{\pi}$ allows to use the same steps as in the classical case with a single prior. By $(i)$, $(ii)$, and $(iv)$, there exists a strategy $\hat{\pi}\in\Pi$ such that
\begin{align*}
\inf_{P\in\mathcal{P}_{\hat{\pi}}}\mathbb{E}_P[u(W_T)]=\mathbb{E}_{P_{\hat{\pi},\hat{\theta}}}[u(W_T)]=\mathbb{E}_{P_{\hat{\pi},\hat{\theta}}}[V(T,W_T,r_T)]=V(0,W_0,r_0).
\end{align*}
By $(i)$ and $(iii)$, for all $\pi\in\Pi$, it holds
\begin{align*}
\inf_{P\in\mathcal{P}_\pi}\mathbb{E}_P[u(W_T)]\leq\mathbb{E}_{P_{\pi,\hat{\theta}}}[u(W_T)]=\mathbb{E}_{P_{\pi,\hat{\theta}}}[V(T,W_T,r_T)]\leq V(0,W_0,r_0).
\end{align*}
Therefore, the strategy $\hat{\pi}$ is an optimal investment strategy and the value of the optimal investment problem is given by $V(0,W_0,r_0)$.
\end{proof}
\noindent In order to find the worst-case prior in Proposition \ref{martingale optimality principle for optimal strategies}, one can again use the martingale optimality principle.
\begin{prp}\label{martingale optimality principle for worst-case priors}
Let $\hat{\pi}\in\Pi$. If there exists a function $V\in C^{1,2,2}([0,T)\times\mathbb{R}_+\times\mathbb{R})$ and a scenario $\hat{\theta}\in\Theta$ such that
\begin{enumerate}
\item[$(i)$] it holds $V(T,W,r)=u(W)$ for all $(W,r)\in\mathbb{R}_+\times\mathbb{R}$,
\item[$(ii)$] the process $(V(t,W_t,r_t))_t$ is a $P_{\hat{\pi},\hat{\theta}}$-martingale,
\item[$(iii)$] for each scenario $\theta\in\Theta$, the process $(V(t,W_t,r_t))_t$ is a $P_{\hat{\pi},\theta}$-submartingale,
\end{enumerate}
then $P_{\hat{\pi},\hat{\theta}}$ is a worst-case prior for $\hat{\pi}$.
\end{prp}
\noindent The proof is similar to the proof of Proposition \ref{martingale optimality principle for optimal strategies}.
\par The martingale optimality principle can be applied by examining the drift of a suitable value function. Under a specific prior $P_{\pi,\theta}$ for a strategy $\pi\in\Pi$ and a scenario $\theta\in\Theta$, the investor's wealth and the short-term interest rate evolve as
\begin{align*}
dW_t&=W_t\big((r_t+\pi_t^\top\lambda_t^\theta)dt+\pi_t^\top\sigma_t^\theta dB_t\big),
\\dr_t&=\kappa(\bar{r}-r_t)dt+(\nu_t^\theta)^\top dB_t,
\end{align*}
where $B$ is a standard Brownian motion.
Under the prior $P_{\pi,\theta}$ the dynamics of an arbitrary function $V\in C^{1,2,2}([0,T)\times\mathbb{R}_+\times\mathbb{R})$ are given by
\begin{align*}
dV(t,W_t,r_t)=\Delta_t^{\pi,\theta}dt+\big(V_W(t,W_t,r_t)W_t\pi_t^\top\sigma_t^\theta+V_r(t,W_t,r_t)(\nu_t^\theta)^\top\big)dB_t,
\end{align*}
where the subscripts denote the partial derivatives of $V$ and the drift $\Delta_t^{\pi,\theta}$ is defined as
\begin{align*}
\Delta_t^{\pi,\theta}:={}&V_t(t,W_t,r_t)+V_W(t,W_t,r_t)W_t(r_t+\pi_t^\top\lambda_t^\theta)+V_r(t,W_t,r_t)\kappa(\bar{r}-r_t)
\\&+\frac{1}{2}V_{WW}(t,W_t,r_t)W_t^2\pi_t^\top\sigma_t^\theta(\sigma_t^\theta)^\top\pi_t+V_{Wr}(t,W_t,r_t)W_t\pi_t^\top\sigma_t^\theta\nu_t^\theta
\\&+\frac{1}{2}V_{rr}(t,W_t,r_t)(\nu_t^\theta)^\top\nu_t^\theta.
\end{align*}
Since the drift term of a process determines if it is a martingale, a supermartingale, or a submartingale, the aim is to find a suitably regular function $V\in C^{1,2,2}([0,T)\times\mathbb{R}_+\times\mathbb{R})$ satisfying the terminal condition $(i)$ in Propositions \ref{martingale optimality principle for optimal strategies} and \ref{martingale optimality principle for worst-case priors}, a strategy $\hat{\pi}\in\Pi$, and a scenario $\hat{\theta}\in\Theta$ such that for all $t\in[0,T]$, it holds
\begin{align}
\Delta_t^{\hat{\pi},\hat{\theta}}&=0,\label{martingale condition}
\\\Delta_t^{\pi,\hat{\theta}}&\leq0\quad\text{for all}\ \pi\in\Pi,\label{supermartingale condition}
\\\Delta_t^{\hat{\pi},\theta}&\geq0\quad\text{for all}\ \theta\in\Theta.\label{submartingale condition}
\end{align}
Provided $V$ is sufficiently regular, condition \eqref{martingale condition} implies that $(V(t,W_t,r_t))_t$ is a $P_{\hat{\pi},\hat{\theta}}$-martingale, condition \eqref{supermartingale condition} implies that $(V(t,W_t,r_t))_t$ is a $P_{\pi,\hat{\theta}}$-supermartingale for each strategy $\pi\in\Pi$, and condition \eqref{submartingale condition} implies that $(V(t,W_t,r_t))_t$ is a $P_{\hat{\pi},\theta}$-submartingale for each scenario $\theta\in\Theta$, i.e., all conditions from Propositions \ref{martingale optimality principle for optimal strategies} and \ref{martingale optimality principle for worst-case priors} are satisfied.
\par One can verify that the strategy, the scenario, and the value function from Theorem \ref{solution to the optimal investment problem} satisfy the conditions on the drift of the value function. Let $\hat{\pi}$, $\hat{\theta}$, and $V$ be defined as in Theorem \ref{solution to the optimal investment problem}. The value function $V$ satisfies condition $(i)$ in Propositions \ref{martingale optimality principle for optimal strategies} and \ref{martingale optimality principle for worst-case priors}. For any strategy $\pi\in\Pi$, the drift of the value function under scenario $\hat{\theta}$ is given by
\begin{align*}
\Delta_t^{\pi,\hat{\theta}}=\exp\Big((1-\gamma)\big(a_0(t)+a_1(t)r_t\big)\Big)W_t^{1-\gamma}f(t,r_t,\pi_t),
\end{align*}
where the function $f:[0,T]\times\mathbb{R}^3\rightarrow\mathbb{R}$ is defined by
\begin{align*}
f(t,r,\pi):={}&a_0'(t)-b'(T-t)r+r+\pi^\top\hat{\lambda}(t)+b(T-t)\kappa(\bar{r}-r)-\frac{1}{2}\gamma\pi^\top\hat{\sigma}(t)\hat{\sigma}(t)^\top\pi
\\&+(1-\gamma)b(T-t)\pi^\top\hat{\sigma}(t)\hat{\nu}+\frac{1}{2}(1-\gamma)b(T-t)^2\hat{\nu}^\top\hat{\nu}
\end{align*}
and the risk premia and the diffusion components under scenario $\hat{\theta}$ are denoted by
\begin{align*}
\hat{\lambda}(t):=\lambda_t^{\hat{\theta}},\quad\hat{\sigma}(t):=\sigma_t^{\hat{\theta}},\quad\hat{\nu}:=\nu_t^{\hat{\theta}}.
\end{align*}
Since the strategy $\hat{\pi}$ can be expressed as
\begin{align*}
\hat{\pi}_t=\frac{1}{\gamma}\big(\hat{\sigma}(t)\hat{\sigma}(t)^\top\big)^{-1}\hat{\lambda}(t)+\frac{1-\gamma}{\gamma}b(T-t)\big(\hat{\sigma}(t)^\top\big)^{-1}\hat{\nu},
\end{align*}
one can check that it maximizes the function $f$, that is,
\begin{align*}
\hat{\pi}_t\in\underset{\pi\in\mathbb{R}^2}{\arg\max}\,f(t,r,\pi)\quad\text{for all}\ (t,r)\in[0,T]\times\mathbb{R}.
\end{align*}
Moreover, the function $a_0$ can be expressed as
\begin{align*}
a_0(t)={}&\frac{1}{2}\frac{1}{\gamma}\int_t^T\hat{\lambda}(u)^\top\big(\hat{\sigma}(u)\hat{\sigma}(u)^\top\big)^{-1}\hat{\lambda}(u)du+\kappa\bar{r}\int_t^Tb(T-u)du
\\&+\frac{1-\gamma}{\gamma}\int_t^T\hat{\nu}^\top\hat{\sigma}(u)^{-1}\hat{\lambda}(u)b(T-u)du+\frac{1}{2}\frac{1-\gamma}{\gamma}\hat{\nu}^\top\hat{\nu}\int_t^Tb(T-u)^2du.
\end{align*}
Inserting the expressions for $\hat{\pi}$ and $a_0$ into $f$ yields
\begin{align*}
f(t,r,\hat{\pi}_t)=0\quad\text{for all}\ (t,r)\in[0,T]\times\mathbb{R}.
\end{align*}
Since $\hat{\pi}$ also maximizes $f$, conditions \eqref{martingale condition} and \eqref{supermartingale condition} hold. Apart from that, for any scenario $\theta\in\Theta$, it holds
\begin{align*}
\Delta_t^{\hat{\pi},\theta}=\exp\Big((1-\gamma)\big(a_0(t)+a_1(t)r_t\big)\Big)W_t^{1-\gamma}g(t,r_t,\theta_t),
\end{align*}
where $g:[0,T]\times\mathbb{R}^6\rightarrow\mathbb{R}$ is defined by
\begin{align*}
g(t,r,\theta)={}&a_0'(t)-b'(T-t)r+r+b(T-t)\kappa(\bar{r}-r)
\\&+\biggl(\frac{1}{\gamma}\frac{1}{1-\hat{\rho}(t)^2}\frac{1}{\hat{\sigma}_r}\biggl(\frac{\hat{\lambda}_B(t)}{\hat{\sigma}_r}+\hat{\rho}(t)\frac{\hat{\lambda}_S}{\hat{\sigma}_S}\biggr)+\frac{\gamma-1}{\gamma}b(T-t)\biggr)\lambda^B
\\&+\frac{1}{\gamma}\frac{1}{1-\hat{\rho}(t)^2}\frac{1}{\hat{\sigma}_S}\biggl(\frac{\hat{\lambda}_S}{\hat{\sigma}_S}+\hat{\rho}(t)\frac{\hat{\lambda}_B(t)}{\hat{\sigma}_r}\biggr)\lambda^S
\\&-\biggl(\frac{1}{2}\biggl(\frac{1}{\gamma}\frac{1}{1-\hat{\rho}(t)^2}\frac{1}{\hat{\sigma}_r}\biggl(\frac{\hat{\lambda}_B(t)}{\hat{\sigma}_r}+\hat{\rho}(t)\frac{\hat{\lambda}_S}{\hat{\sigma}_S}\biggr)\biggr)^2+\frac{1}{2}\frac{\gamma-1}{\gamma}b(T-t)^2\biggr)(\sigma^r)^2
\\&+\frac{1}{\gamma}\frac{1}{1-\hat{\rho}(t)^2}\frac{1}{\hat{\sigma}_r}\biggl(\frac{\hat{\lambda}_B(t)}{\hat{\sigma}_r}+\hat{\rho}(t)\frac{\hat{\lambda}_S}{\hat{\sigma}_S}\biggr)\frac{1}{1-\hat{\rho}(t)^2}\frac{1}{\hat{\sigma}_S}\biggl(\frac{\hat{\lambda}_S}{\hat{\sigma}_S}+\hat{\rho}(t)\frac{\hat{\lambda}_B(t)}{\hat{\sigma}_r}\biggr)\sigma^r\sigma^S\rho
\\&-\frac{1}{2}\frac{1}{\gamma}\biggl(\frac{1}{1-\hat{\rho}(t)^2}\frac{1}{\hat{\sigma}_S}\biggl(\frac{\hat{\lambda}_S}{\hat{\sigma}_S}+\hat{\rho}(t)\frac{\hat{\lambda}_B(t)}{\hat{\sigma}_r}\biggr)\biggr)^2(\sigma^S)^2.
\end{align*}
Distinguishing between the cases $(i)$, $(ii)$, and $(iii)$ from Section \ref{optimal investment}, one can see that the scenario $\hat{\theta}$ minimizes $g$, that is,
\begin{align*}
\hat{\theta}_t\in\underset{\theta\in S(t)}{\arg\min}\,g(t,r,\theta)\quad\text{for all}\ (t,r)\in[0,T]\times\mathbb{R},
\end{align*}
where $S(t):=[\underline{\lambda}_B(t),\overline{\lambda}_B(t)]\times[\underline{\lambda}_S,\overline{\lambda}_S]\times[\underline{\sigma}_r,\overline{\sigma}_r]\times[\underline{\sigma}_S,\overline{\sigma}_S]\times[\underline{\rho},\overline{\rho}]$. Since it also holds
\begin{align*}
g(t,r,\hat{\theta}_t)=f(t,r,\hat{\pi}_t)=0\quad\text{for all}\ (t,r)\in[0,T]\times\mathbb{R},
\end{align*}
condition \eqref{submartingale condition} is also satisfied.
\par It is left to check that the strategy is an admissible investment strategy and that the value function is sufficiently regular to apply the martingale optimality principle. The strategy $\hat{\pi}$ is deterministic and bounded and thus admissible, i.e., $\hat{\pi}\in\Pi$. Since all strategies and scenarios are bounded, one can show that for each $\pi\in\Pi$ and $\theta\in\Theta$,
\begin{align*}
\mathbb{E}\biggl[\int_0^T\Big\Vert V_W(t,W_t,r_t)W_t\pi_t^\top\sigma_t^\theta+V_r(t,W_t,r_t)(\nu_t^\theta)^\top\Big\Vert^2dt\biggr]<\infty,
\end{align*}
which implies that the diffusion part in the dynamics of the value function is sufficiently regular. Therefore, $(V(t,W_t,r_t))_t$ is a $P_{\hat{\pi},\hat{\theta}}$-martingale, $(V(t,W_t,r_t))_t$ is a $P_{\pi,\hat{\theta}}$-supermartingale for all $\pi\in\Pi$, and $(V(t,W_t,r_t))_t$ is a $P_{\hat{\pi},\theta}$-submartingale for all $\theta\in\Theta$ (using the properties of the drift of the value function from above). This completes the proof of Theorem \ref{solution to the optimal investment problem}.

\section{Calibration Neglecting Ambiguity}\label{calibration negelcting ambiguity}
The estimation of the parameters without considering ambiguity is based on a discrete-time approximation of the continuous-time dynamics, similar to the approach of \citet{chanetal1992}, which is the choice of other papers on optimal investment with stochastic interest rates \citep{bajeux-besnainouetal2001,munksorensen2004,sorensen1999}. The discrete-time dynamics of the short-term interest rate, the bond, and the stock, considering only the reference scenario, $\tilde{\theta}$, are
\begin{align*}
r_{t_{i+1}}-r_{t_i}&=\kappa(\bar{r}-r_{t_i})\delta+\sigma_r(B_{t_{i+1}}^1-B_{t_i}^1),
\\\frac{P_{t_{i+1}}^B-P_{t_i}^B}{P_{t_i}^B}-r_{t_i}\delta&=b(\bar{T})\lambda_0^B\delta-b(\bar{T})\sigma_r(B_{t_{i+1}}^1-B_{t_i}^1),
\\\frac{P_{t_{i+1}}^S-P_{t_i}^S}{P_{t_i}^S}-r_{t_i}\delta&=\lambda_S\delta+\sigma_S\rho(B_{t_{i+1}}^1-B_{t_i}^1)+\sigma_S\sqrt{1-\rho^2}(B_{t_{i+1}}^2-B_{t_i}^2)
\end{align*}
for a partition $t_1<...<t_n$ such that $t_{i+1}-t_i=\delta$ for $i=1,...,n-1$. (The dynamics of the bond involve $b(\bar{T})$ instead of $b(\bar{T}-t_i)$, since the data on the bond corresponds to a constant maturity bond.) The actual estimation uses the following representation:
\begin{align*}
\Delta_i^0&=\alpha+\beta R_i^0+\varepsilon_i^1,
\\\Delta_i^B&=\mu_B-b(\bar{T})\varepsilon_i^1,
\\\Delta_i^S&=\mu_S+\varepsilon_i^2,
\end{align*}
where $\Delta_i^0:=r_{t_{i+1}}-r_{t_i}$ is the change in the annualized short-term interest rate over one time step, $R_i^0:=r_{t_i}\delta$ is the risk-free return over one time step (of length $\delta$), $\Delta_i^B:=(P_{t_{i+1}}^B-P_{t_i}^B)/P_{t_i}^B-r_{t_i}\delta$ is the excess return on the bond over the risk-free return over one time step, $\Delta_i^S:=(P_{t_{i+1}}^S-P_{t_i}^S)/P_{t_i}^S-r_{t_i}\delta$ is the excess return on the stock over the risk-free return over one time step, the parameters are defined by $\alpha:=\kappa\bar{r}\delta$, $\beta:=-\kappa$, $\mu_B:=b(\bar{T})\lambda_0^B\delta$, and $\mu_S:=\lambda_S\delta$, and $(\varepsilon_i^1,\varepsilon_i^2)^\top$ for $i=1,...,n-1$ are independent and identically distributed random variables such that
\begin{align*}
\begin{pmatrix}\varepsilon_i^1\\\varepsilon_i^2\end{pmatrix}\sim N(0,\Sigma),\quad\Sigma=\begin{pmatrix}\sigma_r^2\delta&\sigma_r\sigma_S\rho\delta\\\sigma_r\sigma_S\rho\delta&\sigma_S^2\delta\end{pmatrix}.
\end{align*}
\par The estimation uses US data. The data is taken from the Center for Research in Security Prices, LLC (CRSP) and consist of monthly time series ranging from January, 1946, to December, 2020, that is, $\delta=1/12$ and $n=900$. The observations of the return on the risk-free asset, $R_i^0$ for $i=1,...,n$, are returns on 30-day US Treasury bills, which also yield the observations of $\Delta_i^0$ for $i=1,...,n-1$ by taking the difference in the annualized return over each time step. The observations of the excess return on the bond, $\Delta_i^B$ for $i=1,...,n-1$, are the differences in the returns on a 20-year constant maturity US Treasury bond and the returns on 30-day US Treasury bills. The observations of the excess return on the stock, $\Delta_i^S$ for $i=1,...,n-1$, are the differences in returns on the S\&P 500 and the returns on 30-day US Treasury bills.
\par The parameters can then be obtained by using standard estimation techniques. A simple linear regression yields estimates for $\alpha$ and $\beta$, which in turn yield estimates for $\kappa$ and $\bar{r}$, and from the standard deviation of the residuals, one can recover $\sigma_r$. The means of the excess returns on the bond and the stock provide estimates for $\lambda_0^B$ and $\lambda_S$, respectively. Moreover, the standard deviation of the excess returns on the stock and their correlation with the residuals provide estimates for the volatility of the stock, $\sigma_S$, and its correlation with the dynamics of the short-term interest rate, $\rho$, respectively.

\bibliography{C:/Users/juho/Documents/Literature/Literature}
\bibliographystyle{chicago}

\end{document}